\documentclass[12pt]{amsart}

\usepackage{amsmath,amsfonts,amssymb,amsthm,amsopn,cite,mathrsfs}
\usepackage{epsfig,verbatim}
\usepackage{subfigure,color}

\newcommand{\up}{uncertainty principle}
\newcommand{\tfa}{time-frequency analysis}

\newcommand{\tf}{time-frequency}

\newcommand{\fif}{if and only if}
\newcommand{\pss}{phase-space shift}
\newcommand{\tfs}{time-frequency shift}

\newcommand{\psdo}{pseudodifferential operator}

\newtheorem{tm}{Theorem}[section]
\newtheorem{lemma}[tm]{Lemma}
\newtheorem{prop}[tm]{Proposition}
\newtheorem{cor}[tm]{Corollary}

\newcommand{\beqa}{\begin{eqnarray*}}
\newcommand{\eeqa}{\end{eqnarray*}}

\newcommand{\field}[1]{\mathbb{#1}}
\newcommand{\bR}{\field{R}}        
\newcommand{\bN}{\field{N}}        
\newcommand{\bZ}{\field{Z}}      
\newcommand{\bC}{\field{C}}      
       
\newcommand{\bQ}{\field{Q}}

 \def\cG{\mathcal{G}}

\def\rd{\bR^d}

\def\rdd{{\bR^{2d}}}

\def\<{\left<}
\def\>{\right>}

\def\inv{^{-1}}

\def\mv1{M_v^1}

\newcommand{\vs}{\vspace{3 mm}}

\newcommand{\ip}[2]{\ensuremath{\left<#1,#2\right>}}
\newcommand{\sett}[1]{\ensuremath{\left \{ #1 \right \}}}
\newcommand{\abs}[1]{\ensuremath{\left| #1 \right| }}
\newcommand{\Rdst}{{\mathbb{R}^d}}
\newcommand{\Rst}{{\mathbb{R}}}
\newcommand{\Zst}{{\mathbb{Z}}}
\newcommand{\Nst}{{\mathbb{N}}}
\newcommand{\Rtdst}{{\mathbb{R}^{2d}}}

\newcommand{\norm}[1]{\lVert#1\rVert}
\newcommand{\set}[2]{\big\{ \, #1 \, :  \, #2 \, \big\}}

\newcommand{\gablam}{\cG (g,\Lambda  )}

\newcommand{\gab}{\mathcal{G}(g,\alpha , \beta )}
\newcommand{\ltwo}{L^2(\bR )}
\newcommand{\saa}{S^{1,1}} 

\newcommand{\R}{\mathbb{R}}

\newcommand{\Qg}{\mathcal{Q}_g}
\newcommand{\Qphi}{\mathcal{Q}_{\phi}}
\newcommand{\Ag}{\mathcal{A}_g}
\newcommand{\Intalpha}{I_\alpha}
\newcommand{\expg}{\Theta_g^{L}(x,N)}
\newcommand{\expphi}{\Theta_{\phi}^{L}(x,N)}
\newcommand{\newk}{(k_0, \dots, k_{p-1})}

\begin{document}
\begin{abstract}
We investigate the completeness of Gabor systems  with respect to
several classes of window functions  on rational lattices. Our main
results show that the \tf\  shifts of every finite linear
combination of Hermite functions with respect to a rational lattice
are complete in $L^2(\bR )$, thus generalizing a remark of von Neumann
(and proved by Bargmann, Perelomov et al.). An analogous result is
proven for functions  that factor into certain rational functions and the
Gaussian. The results are also interesting from a conceptual point of
view since they show a vast difference between the completeness and
the frame property  of a  Gabor system. In the  terminology of physics we
prove new results about the completeness of coherent state
subsystems. 
\end{abstract}

\title{Completeness of Gabor Systems}
\author{Karlheinz Gr\"ochenig}
\email{karlheinz.groechenig@univie.ac.at}
\author{Antti Haimi}
\email{antti.haimi@math.ntnu.no}

\author{Jos\'e Luis Romero}
\email{jose.luis.romero@univie.ac.at}

\address{Faculty of Mathematics \\
University of Vienna \\
Oskar-Morgenstern-Platz 1 \\
A-1090 Vienna, Austria}
\subjclass[2010]{42C30,42C15,81R30}
\date{}
\keywords{Gabor system, Hermite function, totally positive function,
  coherent state, completeness, frame}
\thanks{K.\ G.\ was
  supported in part by the  project P26273 - N25  of the
Austrian Science Fund (FWF). A.\ H.\ was supported by a Lise Meitner
grant of the Austrian Science Fund (FWF). 
J. L. R. gratefully acknowledges support from an individual Marie
Curie fellowship, under grant PIIF-GA-2012-327063
(EU FP7).}
\maketitle

\section{Introduction}

We study the question when the set of \tfs s 
\begin{equation}
  \label{eq:c1}
\gab = \{ e^{2\pi i \beta l x } g(x-\alpha k) : k,l \in \bZ \}  
\end{equation}
is complete in $\ltwo $, where $g \in \ltwo $   and the
lattice parameters $\alpha , \beta >0$  are fixed. 
The completeness  question arose
first in J.\ von Neumann's treatment of quantum mechanics~\cite{neumann} and
remains relevant in physics and in applied mathematics. The motivation
in signal analysis and \tfa\ comes from   Gabor's fundamental
paper~\cite{gabor46} on information theory. Gabor  tried to expand a
given function (signal) into a series of \tfs s.  Correspondingly, in
mathematical terminology the set $\gab $ is called a Gabor 
system with window function $g$. 
In quantum mechanics the functions $e^{2\pi i \beta l x } g(x-\alpha
k)$ are called \pss s of a (generalized) coherent state, and 
$\gab $ can be
interpreted as a discrete set of coherent states  with respect to the
Heisenberg group   over a lattice $\alpha
\bZ \times \beta \bZ $~\cite{alanga00}. 
In fact,
Perelomov's book~\cite{perelomov86} on coherent states contains several sections devoted
to the ``completeness of coherent state subsystems.''

In the applied mathematics literature of the last 20 years the
interest in Gabor systems has shifted to the frame property,  
mainly for numerical reasons~\cite{feichtinger-strohmer98} and because Gabor frames can
be used to characterize function spaces~\cite{fg97jfa} and to describe \psdo
s~\cite{gro06}. Here
$\gab $ is a (Gabor) frame for $\ltwo $, if there exist constants $A,B>0$ such
that 
\begin{equation}
  \label{eq:c2}
A\norm{f}_2^2 \leq \sum _{k,l \in \bZ } | \langle f, e^{2\pi i \beta l
  \cdot } g(\cdot - \alpha k)\rangle |^2 \leq B \norm{f}_2^2 \qquad
\forall f\in \ltwo \, .
\end{equation}
Clearly \eqref{eq:c2} implies that $\gab $ is complete in $\ltwo $,
but in general the frame property of $\gab$ is much  stronger than its
completeness.
This difference is already present in von Neumann's example $\cG (\phi
, 1, 1)$ where $\phi (x) = e^{-\pi x^2}$ is the Gaussian, i.e.,  the
canonical  coherent state in quantum mechanics,  and $\alpha = \beta =1$. In this case it was
proved 40 years after von Neumann that $\cG (\phi , 1,1)$ is
complete, but not a frame~\cite{bargmann71,perelomov71}.

The main  intuition for the results about the  completeness and the
frame property of $\gab $ is based on the \up . According to the \up\
every physical state $g$, i.e., $g\in \ltwo, \|g\|_2^2 = \int _\bR
|g(x)|^2 \, dx =1$, occupies a cell in phase space (in the \tf\ plane)
of minimal area
one. The \pss\ $e^{2\pi i \beta l x} g(x-\alpha k)$ is located roughly
at position $\alpha k$ and momentum $\beta l$ in phase space $\bR
^2$. Thus, in order to cover the entire phase space with a discrete
set of coherent states $\gab $, we must have necessarily $\alpha \beta
\leq 1$, otherwise there would be gaps in phase space that cannot be
reached by a \pss\ of the form  $e^{2\pi i \beta l x} g(x-\alpha
k)$. The physical intuition has been made mathematically rigorous in
the form of numerous density theorems for Gabor systems: \emph{If
  $\gab $ is complete in $\ltwo $, then necessarily $\alpha \beta \leq
1$}~\cite{daubechies90,RS95,heil07}. The converse holds only for
special window functions. It was already noted in~\cite{bargmann71,perelomov71}
that for the Gaussian $\phi (x) = e^{-\pi x^2}$ the set $\cG (\phi ,
\alpha , \beta )= \{ e^{2\pi i
  \beta l x} e^{- (x-k\alpha )^2} : k,l \in \bZ \}$ is
complete for $\alpha \beta \leq 1$ and incomplete for $\alpha \beta
>1$. In 1992  Lyubarskii\cite{lyub92} and Seip~\cite{seip92} 
strengthened this statement and showed that $\cG (\phi ,
\alpha , \beta )$ is frame, \fif\ $\alpha \beta < 1$. (In fact, they
stated their results for arbitrary \tfs s, not just lattice shifts.)
The  recent work ~\cite{GS13} shows that an analogous result also holds for
the class of  so-called totally positive
functions of finite type. Again, as in the case of the Gaussian,
these functions  possess good \tf\  localization, which in
mathematical terms amounts to additional analyticity properties.

In this paper we return to the  completeness problem for Gabor
systems.  In agreement with the physical description of quantum
states, we will assume 
that the window functions  possess strong localization properties in
the \tf\ plane. Technically,
we will assume that $g$ and its Fourier transform have exponential decay.

Our main results provide a significant generalization of von Neumann
observation on the completeness of coherent states (Gabor systems)  and  may be summarized as follows. 

\begin{tm} \label{tm:intro}
Assume that $g$ factors as  $g(x)= R(x) e^{-\gamma  x^2}$,   where
$\gamma >0$,  $R$ is
either a rational function with no real poles or a finite sum of complex
exponentials $R(x) = \sum _{j=1}^m  c_j  e^{\lambda _j x}$ with $c_j,
\lambda _j \in \bC $. If  $\alpha \beta $ is rational and  $\alpha
\beta \leq 1$,    
then  the
Gabor system $\mathcal{G}(g, \alpha, \beta)$ is 
complete in $L^2(\mathbb{R})$.
\end{tm}

One may say that  the physical intuition works far beyond the canonical
coherent states.
For rational lattices $\alpha \beta \in \bQ$ and windows as in Thm.~\ref{tm:intro}, the Gabor system
(coherent state subsystem) $\gab $ is complete, \fif\ $\alpha \beta
\leq 1$.  Although it seems natural  that Theorem~\ref{tm:intro} can
be extended to 
arbitrary rectangular lattices $\alpha \bZ \times \beta \bZ $ with
$\alpha \beta \leq 1$, we must leave this question  open because  there is no
useful  completeness characterization over 
irrational lattices.

To  put this result in perspective, we note three special cases
(Theorems~\ref{t:pol}, \ref{t:rat}, and \ref{t:expsum} and their
corollaries  in
Section 3). 

(a) Theorem~\ref{tm:intro} covers all Hermite functions $h_n = c_n
e^{\pi x^2} \tfrac{d^n}{dx^n } (e^{-2 \pi x^2})$ (with a suitable
  normalization constant $c_n$ and $n\in \bN $). So far it was known that $\cG (h_n,
  \alpha  , \beta)$ is a frame, if $\alpha \beta <
  \tfrac{1}{n+1}$~\cite{GL09}. However, for  odd
  $n=2m+1$  and  $\alpha
  \beta = 1- \tfrac{1}{N} , N=2, 3, \dots , \infty $,  the Gabor system $\cG (h_{2m+1}, \alpha , \beta
    )$   is not a frame~\cite{LN13}. The achievement of
    Theorem~\ref{tm:intro} is to assert the completeness of $\cG (h_n
    , \alpha , \beta )$ for \emph{all} Hermite functions and
    \emph{all} rational lattices with $\alpha \beta \leq 1$. 

The example of the Hermite functions also indicates the limits of the
physical intuition which does not explain the difference between
completeness (intuition is confirmed perfectly by
Theorem~\ref{tm:intro})  and the frame property (intuition is not
correct). It would be interesting to understand from physical
principles why the frame property can be destroyed by certain
symmetries, as is the case for the odd Hermite functions. 

Let us mention that, after a suitable transformation and choice of
representation,  the Gabor systems $\cG (h_n, \alpha , \beta )$ can
also be interpreted as a coherent state subsystem for the degenerate
Landau levels of a  particle in a constant magnetic
field. Thus our results extend to higher Landau levels what was previously known only for the ground state. The standard complex analytic techniques that can be used to deal with 
the ground state 
do not work with more general Landau levels. 
See~\cite{MR3203099} for a detailed explanation of the
connection between Gabor systems and the Landau levels. 

(b) If $g$ is given by its Fourier transform  $\hat{g}(\xi ) = \prod _{j=1}^m (1 +
i \delta _j \xi )\inv e^{-\gamma \xi ^2}$ for $\delta _j \in \bR $, then $g$ is a totally
positive function by Schoenberg's
characterization~\cite{schoenberg1947totally}. It was conjectured in
~\cite{G14} that the Gabor system $\gab $ for a totally positive
function $g$ generates a frame, \fif\ $\alpha \beta  <1$, but so far
this statement is known to be true only for the class of totally
positive functions of finite type~\cite{GS13}. Theorem~\ref{tm:intro} offers a
similar result for the case of completeness for a complementary
class of totally positive functions and  supports the original
conjecture.

(c) Finally, Theorem~\ref{tm:intro} covers the case when $g$ is a
finite linear combination of \tfs s of the Gaussian 
$\phi (x) = e^{-\pi x^2}$. Again, this is in line with the physical
intuition.

The remainder of the paper will be organized as follows: In Section~2
we recall the abstract characterizations of complete Gabor systems over a
rational lattice due to Zeevi and Zibulski~\cite{zibulski1997analysis}
and derive a
specialized criterium for functions in the Gelfand-Shilov class. We
will juxtapose these characterizations with the corresponding
characterizations of the frame property.   The
comparison of these criteria  reveals the fundamental difference between
the completeness and the frame property and  is
particularly striking for window functions in the Gelfand-Shilov space.  The
completeness property  hinges on the
analyticity of $g$ and $\hat{g}$. It is easy to produce
counter-examples to Theorem~\ref{tm:intro} belonging to 
the Schwartz class, where  $g$ and $\hat{g} $ are
$C^\infty $.

In Section~3 we will prove several versions of
Theorem~\ref{tm:intro}. The main idea is to use our knowledge that
the Gaussian window $\phi (x) = e^{-\pi x^2}$  generates a complete system
$\cG (\phi , \alpha   , \beta ) $ for $\alpha \beta \leq 1$, and  to subsequently show that the
algebraic conditions that guarantee the completeness of $\gab $ for a
factorized function $g(x) = R(x) e^{-\pi x^2}$  are satisfied. This is
a new proof method in Gabor analysis that very likely can be sharpened
and extended. Ultimately we hope that our results will also lead to a
better understanding of several unsolved  problems about Gabor
\emph{frames} that were formulated in~\cite{G14}.

Notation: We write $f(x) \lesssim g(x)$ to say that there exists a
constant independent of $x$, such that $f(x) \leq C g(x)$ for all
$x$. Furthermore, $f(x)  \asymp g(x) $ means that $f(x) \lesssim g(x)$
and $g(x) \lesssim f(x)$. 

\section{Gabor systems}

The translation and modulation operators act on a function $f: \Rst \to \bC$ by
\begin{align*}
&M_b f(x) := e^{2\pi i b x}f(x), \qquad b  \in \Rst,
\\
&T_a f(x) := f(x-a), \qquad a \in \Rst.
\end{align*}
The composition $M_b T_a$ is called a \tfs\ (or  \pss\
in the language of quantum mechanics).
A Gabor system is a collection of \tfs s  of a given function (a
window function in signal analysis, or a
quantum mechanical state).
Given $g\in L^2(\bR)$,  $\alpha , \beta >0$, we define
then 
$$
\gab = \{ M_{\beta l } T_{\alpha k}g: k,l \in \bZ \}.
$$
The frame operator associated with this family is given by
\[ 
Sf = \sum_{k, l \in \mathbb{Z}} \langle f, M_{l \beta} T_{k \alpha} g \rangle M_{\beta l} T_{\alpha k } g.
\]
Under mild assumptions on $g$  the frame operator is
bounded~\cite{walnut92,MR1843717}. The   standard assumption is that $g$ belongs to the modulation space
\begin{align*}
M^1(\Rst) := \set{f \in L^2(\Rst)}{\int_{\Rst^2} \abs{\ip{f}{M_\xi T_x
      f}} dx d\xi  < +\infty},
\end{align*}
also known as the \emph{Feichtinger algebra}. If  $g \in M^1(\Rst)$,
then $S$ is bounded
on $L^2(\Rst)$. 
A Gabor system $\gab$ is called a \emph{frame} if the frame operator
$S$ is invertible on $L^2(\Rst)$. This is equivalent to  the frame inequalities
\eqref{eq:c2}.

We say that $\gab$ is complete (in $L^2$) if the linear span of $\gab$ is a dense subspace of 
$L^2(\Rst)$. Equivalently, $S$ is one-to-one on $\ltwo $. 

While the completeness of $\gab$ means that any function $f \in L^2(\Rst)$ can be approximated
by linear combinations of elements of $\gab$, the frame property implies the existence of a convergent expansion
\begin{align*}
f = \sum_{k, l \in \mathbb{Z}} a_{k,l} M_{\beta l} T_{\alpha k } g,
\end{align*}
with $\norm{a}_2 \asymp \norm{f}_2$.  For  the Gaussian $\phi (t) :=
e^{-\pi t^2}$, 
$\cG (\phi , \alpha, \beta ) $ is complete if and only if $\alpha \beta \leq 1$
and $\cG (\phi , \alpha, \beta ) $  is a frame if and only if $\alpha \beta < 1$ 
\cite{lyub92,perelomov71,seip92}.

\subsection{Completeness criteria for rational lattices}
Our starting point is the characterization of the completeness of a Gabor system
over a rational lattice by means of the Zak transform, due to Zeevi and
Zibulski \cite{zibulski1997analysis}.

We write $\alpha \beta = p/q$ for relatively prime positive integers
$p$ and $q$ and assume that $p\leq q$.  The Zak transform of $f \in L^2(\mathbb{R})$ with respect to a
parameter $ \alpha \in \mathbb{R}$ is 
\[
Z_\alpha f(x, \xi) = \sum_{k \in \mathbb{Z}} f(x- \alpha k) e^{2 \pi i
  \alpha k \xi }.
\]
By quasi-periodicity  this
function is completely determined by its values on the rectangle
$\Intalpha = [0,\alpha ] \times [0, 1/\alpha ]$. Furthermore, the Zak transform is a unitary isomorphism
from 
$L^2(\mathbb{R})$ to $L^2(\Intalpha)$.

A note on terminology:  In signal processing  $Z_\alpha $ is called  the Zak
transform in reference to one of its first applications by  J.\
Zak~\cite{zak67},  in  solid states physics $Z_\alpha $ is usually
called the 
Bloch-Floquet transform ~\cite{kuch93}, and  in harmonic analysis it is
often called the Weil-Brezin transform~\cite{folland89}.

 We will also need the vector-valued version 
\[
\overrightarrow{Z_{\alpha}}f(x, \xi) = \big( Z_\alpha f(x+ \tfrac{\alpha}{p} r, \xi) \big)_{r=0}^{p-1}, \quad x,\xi 
\in 
[0, 
\alpha/p) \times [0, 1/\alpha), 
\] 
which is unitary  from $L^2(\mathbb{R})$ onto the  space $L^2 \big( [0, \alpha/p) \times [0, 
1/\alpha), \mathbb{C}^p \big)$ of vector-valued functions.

For rational lattices $\alpha \bZ \times \beta \bZ $ with $\alpha
\beta \in \bQ $, the characterizations of the frame property and of the completeness of $\gab$
are in terms of the spectrum of a certain matrix containing
Zak transforms of $g$.
Let $\Qg(x,\xi )$  be the $p \times q   $-matrix with entries 
\[
\Qg(x,  \xi)_{jk}  = Z_\alpha g (x+ \frac{\alpha j}{p}, \xi - \beta k)
e^{2 \pi i j k /q}  \qquad   j=0, \dots , p-1, k= 0, \dots
,  q-1. 
\]
and $\Ag(x,\xi)$ be the corresponding $p\times p$ square matrix 
\[
\Ag(x,\xi) = \Qg(x,\xi) \Qg(x, \xi)^* \, .
\]

Then the  frame operator has the representation 
\begin{equation} \label{frameopzak}
\overrightarrow{Z_{\alpha}}Sf(x, \xi)= \alpha^d \Ag(x,\xi) \overrightarrow{Z_{\alpha}}f(x, \xi), \quad x,\xi \in [0, 
\alpha/p) 
\times [0, 1/\alpha) , 
\end{equation}
which is well defined for  $g \in M^1(\Rst)$. 
The following characterization of completeness is 
 due to Zeevi and Zibulski~\cite[Theorem 2]{zibulski1997analysis}. (We use a slightly different setup that follows
\cite[Chapter 8]{MR1843717}.)
\begin{lemma}
  \label{l:comp1}
Let $g\in L^2(\Rst)$ and $\alpha \beta = p/q \in \bQ$. Then following
are equivalent.

(i) The Gabor system $\gab $ is complete in $L^2(\bR )$. 

(ii) $\det \Ag(x,\xi)\neq 0$, for almost every $(x,\xi) \in \Rst^2$.

(iii) The matrix $\Qg(x,\xi )$ has full rank $p$  for almost every $(x,\xi) \in \Rst^2$.
\end{lemma}

It is instructive to compare this characterization of completeness
with the corresponding characterization of the  frame property
in~\cite[Thm. 4]{zibulski1997analysis}: 
$\gab$ is a frame, \fif\ there exist $0<\delta \leq \Delta $ such that
$\delta \leq \abs{\det
  \Ag(x,\xi)} \leq \Delta$, for almost every $(x,\xi) \in \Rst^2$. 
When $g \in M^1(\Rst)$, the continuity and quasi-periodicity of $Z_g$
yield the following simple characterization
of the frame property.

\begin{lemma}
  \label{l:frame1}
Let $g\in M^1(\bR )$ and $\alpha \beta = p/q \in \bQ$. Then following
are equivalent: 

(i) The Gabor system $\gab $ is a frame  in $L^2(\bR )$. 

(ii)  $\det \Ag(x,\xi)\neq 0$, for all $(x,\xi) \in \Rst^2$. 

(iii) The matrix $\Qg(x,\xi )$ has full rank for all $(x,\xi) \in \Rst^2$.  
\end{lemma}

Note the subtle difference in conditions (ii)! A single zero of $\det
\Ag $ may destroy the frame property of $\gab $. 

We mention that in both cases one may write a formal reconstruction of
$f\in \ltwo $ from the correlations $\langle f, M_{\beta l } T_{\alpha
  k}g\rangle , k,l \in \bZ $ by inverting \eqref{frameopzak} as
follows:
\begin{equation}
  \label{eq:c7}
  f = \alpha ^{-d} \, \overrightarrow{Z_{\alpha}} \inv \Big( \Ag (x,\xi )\inv
  \overrightarrow{Z_{\alpha}}S f(x, \xi) \Big) \, .  
\end{equation}
If $\gab $ is a frame, then this reconstruction is stable, whereas for
a complete Gabor system this reconstruction may lead to instabilities
on certain subspaces of $\ltwo $, because $\Ag $ is not invertible
everywhere. For the classical coherent states $\phi (x) = e^{-\pi
  x^2}$ and $\alpha = \beta =1$ explicit reconstruction formulas are
known, see e.g.~\cite{Ner06}.

\subsection{Windows in the Gelfand-Shilov class}

For windows in the Gelfand-Shilov class $S^{1,1}$ the characterization
of completeness can be reformulated in a useful way, which  we now describe.

We say that a function $g:\Rst \to \bC$ is in the Gelfand-Shilov class $\saa(\Rst)$, if
$g$ and its Fourier transform $\hat{g}$ have exponential decay, i.e.,
$$
|g(x)| \lesssim   e^{-a|x|} \qquad  \forall x\in \R  $$
and 
\begin{equation}
  \label{eq:c11}
|\hat{g}(\xi)| \lesssim  e^{-b|\xi|} \qquad  \forall \xi \in \R \, 
\end{equation}
for some decay constants $a,b>0$. 
This is not the standard definition but a simpler equivalent condition
due to ~\cite{MR1265470}. 
In particular, every function of the form $g=R \phi $, where $R$ is a
rational function with no poles on the real axis and $\phi (x) =
e^{-\pi x^2}$, belongs to  the
Gelfand-Shilov class $\saa $.  Likewise,  $\saa$ includes functions
of the form 
\[
g(x)= \sum_{k=1}^n c_k e^{2\pi i b_k x} \phi(x-a_k),
\]
where $a_1,\ldots,a_n,b_1,\ldots,b_n \in \Rst$.

The assumption $g\in \saa $ implies that $g$ is real analytic. More precisely,
by the theorem of Paley-Wiener every $g\in \saa $ extends to a
function that is analytic on the strip $S_b = \{ z\in \bC : |\Im z| <
b\}$, where $b$ is the decay constant in \eqref{eq:c11}. Moreover, for every $b' \in (0,b)$ there exists a constant
$C_{b'}$ such that
\begin{align*}
\abs{g(x+iy)} \leq C_{b'} e^{-a|x|} \qquad  x,y\in \R
\mbox{ with } \abs{y} \leq b'.
\end{align*}
See for example \cite[Theorem 3.9]{MR1322917}

We now observe that the Zak transform of $g\in \saa $ can also be extended to an analytic function
of two variables.

\begin{lemma}
\label{lemma_zak_analytic}
  If $g\in \saa(\Rst)$, then the Zak transform $Z_\alpha g$ can be extended
  to an analytic function on $S_b \times S_b \subseteq \bC ^2$ for some $b>0$.
\end{lemma}
\begin{proof}
Since $g\in \saa$, there exists $b>0$ such that $g$ extends analytically to $S_b$ and satisfies the uniform 
decay estimate:
\begin{align*}
\abs{g(x+iy)} \lesssim e^{-2 \pi a \abs{x}} \qquad x \in \Rst, \abs{y} < b,
\end{align*}
for some $a>0$.
Without loss of generality, we assume that $b < a$. Let us show that $Z_\alpha g$ extends analytically to $S_b \times 
S_b$. It suffices to show that the series
\begin{align*}
Z_\alpha g(z, w) = \sum_{k \in \mathbb{Z}} g(z- \alpha k) e^{2 \pi i
  \alpha w k} 
\end{align*}
converges uniformly and absolutely on compact subsets of $S_b \times S_b$. Fix 
$C>0$ and let $(z,w) \in \bC^2$ with 
$\abs{z}, \abs{w} \leq C$ and $\abs{\Im z}, \abs{\Im w} \leq b$. Then
\begin{align*}
&\sum_{k \in \mathbb{Z}} \abs{g(z- \alpha k) e^{2 \pi i \alpha w k}}
\lesssim 
\sum_{k \in \mathbb{Z}} e^{-2 \pi a \abs{z- \alpha k}} e^{2 \pi \alpha
  \abs{\Im w} \abs{k}} \\
 & \qquad \leq
\sum_{k \in \mathbb{Z}}
e^{-2\pi (a \alpha \abs{k} - a\abs{z} - \alpha b \abs{k})}
\leq 
e^{2\pi aC} \sum_{k \in \mathbb{Z}}
e^{-2\pi \alpha (a-b) \abs{k}}
<+\infty , 
 \end{align*}
since $a> b $. This completes the proof.
\end{proof}

Using the analyticity of the Zak transform, we now  obtain the following characterization for
the completeness of a Gabor system with windows $g \in \saa$.

\begin{prop} \label{propan}
  Let $g\in \saa(\Rst)$ and $\alpha \beta  =p/q \in \bQ $, $\alpha \beta  \leq 1$.
Then the following are equivalent.

(i) $\gab $ fails to be complete. 

(ii) $\det \Ag(x,\xi) =0 $ for all $x  , \xi \in \bR $. 

(iii) $\Qg(x,\xi )$ has rank $ < p$ for all $x,\xi \in \bR $, in other
words,  all $p\times p$ submatrices of $\Qg(x,\xi )$ have determinant
zero. 
\end{prop}

\begin{proof}
The equivalence between $(ii)$ and $(iii)$ is clear. In addition, by Lemma
\ref{l:comp1}, $(iii)$ implies $(i)$. Let us show that $(i)$ implies
$(iii)$. Assume that $\gab$ is incomplete,  and  let $l_0, \dots, l_{p-1} \in
\{ 0, \dots , q-1\}$ be distinct indices, and
denote by $\Qg^{l_0, \dots, l_{p-1}}$ the  $p \times p$
submatrix of $\Qg(x,\xi ) $ with columns  $l_0, \dots, l_{p-1}$. Let
\begin{align*}
E := \set{(x,\xi) \in \Rst^2}{\det 
\Qg^{l_0, \dots, l_{p-1}}(x,\xi)=0}. 
\end{align*}
We will  show that $E=\Rst^2$. By Lemma \ref{l:comp1}, we know that $E$ has positive Lebesgue measure,
since, otherwise, $\Qg(x,\xi)$ would have full rank almost everywhere.
In addition, by Lemma \ref{lemma_zak_analytic}, there exists $b>0$ such that
$Z_\alpha$ extends analytically to $S_b \times S_b$ and, therefore, so does $\det \Qg^{l_0, \dots, l_{p-1}}$.
Hence, $E$ is the zero set of a real analytic function on $\Rst^2$, and, having positive measure, it must be equal to 
$\Rst^2$. We provide a short argument for this (known)  fact.

By the continuity of $\det \Qg^{l_0, \dots, l_{p-1}}$, we know that $E$ is a closed set.
Hence, in order to prove that $E=\Rst^2$, it suffices to show that $|\Rst^2 \setminus E| = 0$.
Given $x \in \Rst$, the section
\begin{align*}
E_x := \set{\xi \in \Rst}{(x,\xi) \in E}
\end{align*}
is the zero set of the function $\det \Qg^{l_0, \dots, l_{p-1}}(x,\cdot):\Rst \to \bC$, which admits an 
analytic extension to $S_b$. Hence, $E_x$ is either $\Rst$ or has measure zero.
Therefore, $E=(X \times \Rst) \cup N_1$, for some measurable sets $X,N_1 \subseteq \Rst$ with $|N_1|=0$ and $|X|>0$.
A similar argument, with the roles of $x$ and $\xi $ interchanged,
shows  that $E=(\Rst\times Y) \cup N_2$, with $|Y|>0$ and
$|N_2|=0$. The conditions on $E$ imply that 
$1_E(x,y)=1_X(x)=1_Y(y)$  for almost all $(x,y) \in \Rst^2$. Since  $\abs{E}>0$,
this is only possible if $1_E \equiv 1$ a.e. Therefore, $|\Rst^2 \setminus E| = 0$, as desired.
\end{proof}

Motivated by Proposition \ref{propan}, we compute explicitly the determinant of $ \Qg^{l_0,
  \dots, l_{p-1}} (x,\xi ) $ for a given selection of columns $L \equiv \{ l_0 , \dots , l_{p-1}\} \subseteq
\{ 0, \dots , q-1\}$.

Let $\mathrm{Perm} (l_0,
  \dots, l_{p-1})$ denote the group of all permutations of the chosen
  columns. Then 
\begin{multline*}
\det \Qg^{l_0, \dots , l_{p-1}}(x,\xi )  = \sum_{\sigma \in
  \mathrm{Perm}(l_0, \dots, l_{p-1})} (-1)^{\mathrm{sgn} (\sigma)} 
\prod_{j=0}^{p-1} Z_{\alpha}g(x+ \frac{\alpha j}{p}, \xi -  \beta \sigma(l_j)) e^{2 \pi i j \sigma(l_j)/q} \\
= \sum_{\sigma \in \mathrm{Perm}(l_0, \dots, l_{p-1})}
(-1)^{\mathrm{sgn} (\sigma)}  
\sum_{k_0, \dots, k_{p-1} \in \mathbb{Z}} \prod_{j=0}^{p-1} g(x +
\frac{\alpha j}{p} - \alpha k_j) \times \\
e^{2 \pi i  \sum_{j=0}^{p-1}\big[(\xi - \beta \sigma(l_j))\alpha k_j + j \sigma(l_j)/q \big]}=0.
\end{multline*}
We sum over the permutations $\sigma$ first and denote the resulting
coefficients by 
\begin{align}
\label{eq_def_ck}
c(k_0, \dots, k_{p-1}):= 
\sum_{\sigma \in
  \mathrm{Perm}(l_0, \dots, l_{p-1})} (-1)^{\mathrm{sgn} (\sigma)} e^{2 \pi i
  \sum_{j=0}^{p-1}\big[ - p \sigma(l_j)k_j/q + j \sigma(l_j)/q \big]}
. 
\end{align}
Then 
\begin{align*}
\det \Qg^{l_0, \dots , l_{p-1}}(x,\xi )  &=  \sum_{k_0, \dots, k_{p-1} \in \mathbb{Z}} \prod_{j=0}^{p-1} g(x +
\frac{\alpha j}{p} - \alpha k_j)c(k_0, \dots , k_{p-1}) e^{2\pi i
  \alpha (k_0
  + \dots + k_{p-1} ) \xi } \\
&= \sum _{N\in \bZ } \Big( \sum_{k_0 +  \dots +  k_{p-1}  = N} \prod_{j=0}^{p-1} g(x +
\frac{\alpha j}{p} - \alpha k_j)c(k_0, \dots , k_{p-1}) \Big) e^{2\pi
  i N  \xi } \, .
\end{align*}
Thus for fixed $x \in [0, \alpha ]$ the determinant $\det
\Qg^{l_0, \dots , l_{p-1}}(x,\xi ) $ is a Fourier series
with coefficients 
\begin{equation} \label{sum2}
\expg:=\sum_{k_0 + \dots+ k_{p-1} = N} \prod_{j=0}^{p-1} g\big(x+ \alpha j/p
- \alpha k_j\big)c(k_0, \dots, k_{p-1}).
\end{equation}
Note that the
selection of columns of $\Qg ^{l_0, \dots , l_{p-1}}(x,\xi )$ enters
only in the coefficients $c(k_0, \dots, k_{p-1})$. 

We now state one more reformulation of Lemma~\ref{l:comp1}. 

\begin{lemma} \label{l:comsaa}
Let $g \in \saa(\Rst) $ and $\alpha \beta = p/q \in \bQ $, $\alpha \beta \leq 1$.  Then the following
are equivalent.

(i) The Gabor system $\gab $ is incomplete in $\ltwo $. 

(ii) For all choices of subsets $L\equiv \{l_0, \dots, l_{p-1}\} \subset \{0, \dots, q-1 \}$, 
all $x\in \bR$ and all $N \in \Zst$, we have $\expg=0$.
\end{lemma}

\subsection{Complete Gabor Systems versus Gabor Frames}
Although the formulations of Lemma~\ref{l:comp1} and \ref{l:frame1}
 are deceptively similar, the
difference between completeness and the frame property is dramatic. To
highlight this difference, we show that  the spanning properties of a
Gabor family over a lattice are extremely stable and cannot  
be modified  by naive changes. In order to state the result precisely
(for arbitrary dimension), we recall
that the \emph{lower Beurling density} of a set $N \subseteq \Rst^d$
is given by 
\begin{align*}
D^{-}(N) := \liminf_{R \longrightarrow +\infty} \inf_{x \in \Rst^d}
\frac{\#(N \cap B_R(x))}{\abs{B_R(x)}} \, ,
\end{align*}
where $B_R(x)$ denotes the ball of radius $R$ centered at $x\in \rd
$. Hence $D^{-}(N)=0$ if and only if $N$ contains arbitrarily large holes.

\begin{prop}
\label{prop_no_sup}
  Assume that $g\in M^1(\Rdst)$, $\Lambda \subseteq \Rtdst$ is a lattice  
  and  $\gablam $ is not a frame. If $N \subseteq \rdd $ is a set such that $D^-(N) = 0$, then $\cG (g, \Lambda \cup
  N)$ fails to be a frame. 
\end{prop}
\begin{proof}
Assume on the contrary that $\cG (g, \Lambda \cup  N)$ is a frame. The proof is based on Theorem 5.1 from 
\cite{grorro15}, which implies that
for every set $\Gamma$ that is a weak limit of translates of $\Lambda \cup  N$, $\cG (g, \Gamma)$ is also a frame. (We 
refer the reader to \cite{grorro15} for the precise definition of weak convergence of sets.)

Since $D^-(N) = 0$, $N$ contains arbitrarily large holes centered at points of $\Lambda$.
 For every  $n \in \Nst$  
there exists $\lambda_n \in \Lambda$ such that $N \cap B_n(\lambda_n) = \emptyset$. This implies that the sequence
of translates $\sett{(\Lambda \cup N)-\lambda_n: n \in \Nst}$ converges weakly to $\Lambda$. Indeed
\begin{align*}
\left(
(\Lambda \cup N)-\lambda_n
\right) \cap B_n(0) =
\left(
\Lambda \cup (N-\lambda_n)
\right) \cap B_n(0) = \Lambda \cap B_n(0).
\end{align*}
 By \cite[Theorem 5.1]{grorro15}  $\cG (g, \Lambda)$ is a frame, contradicting our assumptions.
\end{proof}

Proposition \ref{prop_no_sup} means that in order to extend a complete Gabor family (that is not already a frame) 
into a (non-uniform) Gabor frame, we need to add a set of strictly positive density. This complements the 
results of Balan, Casazza, Heil and Landau that also stress the strong rigidity of the frame property:
for a Gabor frame $\gablam$ with window $g \in M^1$, it is always possible to remove a certain infinite subset 
$\Lambda' \subseteq \Lambda$ while preserving the frame property. Moreover, 
it is possible to choose $\Lambda'$ so that the density of the remaining set $\Lambda \setminus \Lambda'$ is 
arbitrarily close to 1 \cite{MR2837145}.

\section{Explicit Completeness Theorems}

In this section we investigate the completeness  of Gabor systems $\gab $
for several general,  explicit classes of window functions  in $\saa
$ on a rectangular rational lattice.  
In all cases we will make use of the fact that the Gabor system $\cG
(\phi , \alpha ,  \beta )$ for $\phi (x) = e^{-\pi x^2}$ is complete
when $\alpha \beta \leq 1$ and then apply the algebraic
characterization of Lemma~\ref{l:comsaa}.

 We start with the
class of  windows that factor into a polynomial and the Gaussian. 
We first  prove that,  for an arbitrary polynomial $P$, the
function  $g(x)=P(x)e^{-\pi x^2}$ generates a complete 
system when $\alpha \beta$ is rational and $  \alpha \beta \leq 1$.

\begin{tm} \label{t:pol}
Let $g(x)= P(x) e^{-\pi x^2}$ with a non-zero polynomial  $P$ and $\alpha \beta \in \bQ $, $\alpha \beta \leq 1$.   
Then  the
Gabor system $\mathcal{G}(g, \alpha, \beta)$ is 
complete in $L^2(\mathbb{R})$.
\end{tm}
\begin{proof}
We set $d:=\deg(P)$ and assume without loss of generality that the leading coefficient of $P$ is $1$. As before, we 
write $\alpha \beta = p/q$ for relatively prime positive integers.

We recall the fundamental fact  that the Gaussian $\phi (x)=e^{-\pi
  x^2}$ generates a complete Gabor system $\cG (\phi , \alpha   , \beta )$, \fif\   $\alpha \beta \leq
1$~\cite{bargmann71,perelomov71}.  We will show that  this fact implies that $\cG (g, \alpha, \beta )$
with $g = P \phi $ is also complete. 
More precisely,  we will verify the conditions of
Proposition~\ref{propan}  and show that there
exist column indices  
$L\equiv\{l_0, \dots, l_{p-1}\}$ such that the determinant of $\Qphi ^{l_0, \dots,
  l_{p-1}}$ does not vanish identically.

We argue by contradiction  and  assume that $\gab $ is not
complete. By Lemma \ref{l:comsaa},
\begin{equation}
\label{sum2b}
\expg=\sum_{k_0 + \dots+ k_{p-1} = N} \prod_{j=0}^{p-1} g\big(x+ \frac{\alpha j}{p} - \alpha k_j\big)c(k_0, \dots, 
k_{p-1}) =0 
\end{equation}
for all $x \in \mathbb{R}$ and all $N \in \mathbb{Z}$ where the
coefficients 
$c(k_0, \dots, k_{p-1})$ are given by \eqref{eq_def_ck}.

\textbf{Step 1.} We first  compute the expression $\expphi$ for the
Gaussian $\phi (x) = e^{-\pi x^2}$ in place of $g$.  The exponent of
each product is
\begin{align*}
\lefteqn{-\frac{1}{\pi} \log  \prod_{j=0}^{p-1} \phi \big(x+
  \frac{\alpha j}{p} -
\alpha k_j \big)} \\
&=  \sum_{j=0}^{p-2} \big(x+ \frac{\alpha j}{p}- \alpha k_j \big)^2 + \big( x+
\frac{\alpha(p-1)}{p}  - \alpha(N- k_0-\dots- k_{p-2} ) \big)^2 \\ 
&= px^2 + 2x \alpha \big[\sum_{j=0}^{p-2} (\frac{j}{p}-k_j ) + \frac{p-1}{p} - N + \sum_{j=0}^{p-2} k_j \big] \\
& \qquad \quad  + \alpha^2\Big[ \sum_{j=0}^{p-2} \Big(\frac{j}{p} - k_j \Big )^2 +
\Big(\frac{p-1}{p} -N+ \sum_{j=0}^{p-2} k_j \Big)^2 \Big] \\
& = px^2 +\frac{2 x \alpha}{p} \left(\sum_{j=0}^{p-1}j\right) - 2x \alpha N 
+ \alpha^2\bigg[  \sum_{j=0}^{p-2} \Big(\frac{j}{p} - k_j \Big)^2 +
\Big(\frac{p-1}{p} -N+ \sum_{j=0}^{p-2} k_j \Big)^2  \bigg] \, .
\end{align*}
Since $\cG (\phi, \alpha  , \beta )$ is complete for $\alpha \beta
\leq 1$,
there exists some
$N\in \bZ$, such that  the quantity 
\begin{align}
   e^{-\pi (p x^2 + \alpha (p-1)x - 2N\alpha x)} &  \, \sum_{k_0, \dots, k_{p-2} \in \mathbb{Z}}
\exp \Big(- \pi \alpha^2 \bigg[ \sum_{j=0}^{p-2} \bigg(\frac{j}{p} - k_j
\bigg)^2 +  \bigg(\frac{p-1}{p} -N+ \sum_{j=0}^{p-2}k_j
\bigg)^2\bigg] \Big)  \times \notag \\ 
& \quad c(k_0, \dots, k_{p-2}, N-k_0 -\dots- k_{p-2}) \not \equiv
0\, .    \label{eq:c20}
\end{align}

Thus at least  one of the coefficients
\begin{align}
  s_0(N)& =   \sum_{k_0, \dots, k_{p-2} \in \mathbb{Z}}
\exp \Big(- \pi \alpha^2 \bigg[ \sum_{j=0}^{p-2} \bigg(j/p - k_j
\bigg)^2 +  \bigg(\frac{p-1}{p} -N+ \sum_{j=0}^{p-2}k_j
\bigg)^2\bigg] \Big)  \times  \notag 
 \\ 
& \qquad c(k_0, \dots, k_{p-2}, N-k_0 -\dots- k_{p-2}) \neq  0 \,  
\label{eq:c21}
\end{align}
must be non-zero. 

\textbf{Step 2.}
We next evaluate $\expg$ for $g(x) = P(x) e^{-\pi x^2}$. Since
$P$ is a polynomial of degree $d$ with leading coefficient $1$, the
product 
$\prod_{j=0}^{p-1} P \big(x+ \alpha j/p -
\alpha k_j \big)$ is a polynomial of degree $dp$. The coefficient of
$x^{dp} $ is one, whereas the coefficients of the lower order terms
depend on $(k_0, \dots , k_{p-1})$ in a complicated way. What matters
is that the leading coefficient does not depend on $(k_0, \dots ,
k_{p-1})$. After summing over $(k_0, \dots , k_{p-1})\in \bZ ^p$ with
$k_0 + \dots +  k_{p-1} = N$, we obtain  from \eqref{sum2b} and the
calculation in Step~1 that 
\begin{multline} \label{sum3}
e^{-\pi \big[ px^2 +  x \alpha(p-1) - 2x \alpha N
  \big]} \bigg[ x^{pd} s_0(N) + x^{pd-1} s_1(N) + \dots+ x
s_{pd-1}(N) + s_{pd}(N) \bigg] = 0, 
\end{multline}
for some coefficients $s_k(N) \in \bC$,
all $N \in \mathbb{Z}$ and $x \in \Rst$. The coefficient of the leading
term is exactly $s_0(N)$ from \eqref{eq:c21} obtained for the Gaussian. 
Since we have assumed that $\gab $ is incomplete, \eqref{sum3}
vanishes identically for all $N\in \bZ $ by Lemma~\ref{l:comsaa}, and thus  all
coefficients $s_m(N), m=0, \dots, pd,$ must vanish. In particular,
$s_0(N) = 0$ for all $N\in \bZ $, contradicting \eqref{eq:c21}.

Altogether we have proved that $\gab $ is complete for rational
$\alpha \beta \leq 1$.  
\end{proof}

We single out the special case of the Hermite functions $h_n$ defined
by $h_n(x) =  
e^{\pi x^2} \tfrac{d^n}{dx^n } (e^{-2 \pi x^2}) $. 

\begin{cor} \label{cor:herm}
If $\alpha \beta \in \bQ$ and 
  $\alpha \beta \leq 1$,  then $\cG (h_n, \alpha, \beta)$  is complete in $L^2(\bR )$.
\end{cor}

It is known that for $\alpha \beta = 1 - 1/N$ for $N=2,3, \dots ,
\infty $    the  Gabor system
$\cG (h_{2n+1}, \alpha , \beta )$ cannot be a frame\cite{LN13}.
Furthermore, the frame property fails for $\cG (h_{4m+2}, \alpha,
\beta )$ and $\cG (h_{4m+3}, \alpha,
\beta )$ over certain rational lattices~\cite{lem15}.
 Corollary~\ref{cor:herm} shows that for all Hermite windows
and the lattice parameters where the frame property has been shown to fail, the weaker property of completeness still 
holds.

The result concerning Hermite functions has also  an interpretation in
terms of complex analysis. We define the  
\emph{Bargmann-Fock spaces of polyanalytic functions} as 
\[
A^2_n := \set{f:\bC\to\bC}{\int_{\bC} |f(z)|^2 e^{- \pi |z|^2} \mathrm{d} A(z)
<\infty, \bar \partial^n f =0},
\]
where $\mathrm{d} A(z)$ is the area measure on the complex plane and $\bar \partial = \frac12(\partial_x+ i 
\partial_y)$. The corresponding orthogonal difference spaces  
\[
\delta A^2_n := A^2_n \ominus A^2_{n-1}
\]
are called \emph{true polyanalytic Bargmann-Fock spaces}. These appear naturally in quantum mechanics 
\cite{MR3203099, haimi2013polyanalytic}, where they are sometimes
called the Landau levels. 
As in the case of the Gaussian $n=0$, one verifies that  
$\mathcal{G}(h_n, \alpha, \beta)$ is a complete system if and only if $\alpha \mathbb{Z} \times \beta 
\mathbb{Z}$ is a uniqueness set for $\delta A^2_n$ (see \cite{MR2672228,MR3203099}).  

According to the previous corollary, any rectangular lattice of rational density larger or equal to $1$ is a 
uniqueness set for $\delta A^2_q$. It would be interesting to obtain an alternative proof of this fact using complex 
analysis. The standard techniques for analytic functions do not seem to work in the polyanalytic setting.

Next we consider windows of the form $g(x) = R(x)  \, e^{-\pi  x^2}$, where $R= P/Q$ is  a rational function with two  
polynomials $P$ and $Q$. We may assume without loss of generality that the
leading coefficients of $P$ and $Q$ are $1$.  For well-posedness  we must assume
that $R$ does not have any poles  on the real axis. Then, in fact, the 
poles of $g$ are outside a strip $S_a = \{ z\in \bC : |\Im z| < a\}$
for some $a>0$. It is then  easy to see that $g\in \saa $.

Before stating a completeness theorem for windows of this type, we
formulate a fact about rational functions required later. 
\begin{lemma}
\label{lemma_poly}
Let  $P,Q 
\in \mathbb{C}[X]$ be two monic polynomials  of degree $n$, and let
\begin{align*}
M := \max\sett{\abs{z}: z \in \bC, P(z)=0\mbox{ or }Q(z)=0}.
\end{align*}
Then   there exists a constant $C_n >0$ depending only on the degree
$n$, but not on $M$,  such that 
\begin{align*}
\abs{\frac{P(x)}{Q(x)} - 1} \leq 
\frac{C_n M^n}{\abs{x}},
\qquad
\text{ for all } \abs{x} \geq \max\{2 M,1\}.
\end{align*}
\end{lemma}
\begin{proof}
Let $\abs{x} \geq 2 M$ and let $Q(z) := \prod_{k=1}^n (z-z_k)$ with $z_k \in \bC$. Then
\begin{align*}
\abs{x-z_k} \geq \abs{x}-\abs{z_k} \geq \abs{x}-M \geq \frac{\abs{x}}{2},
\end{align*}
and therefore $\abs{Q(x)} \geq 2^{-n} \abs{x}^n$.

Writing  $P(x) = x^n + \sum_{k=0}^{n-1} a_k x^k$ and
$Q(x) = x^n + \sum_{k=0}^{n-1} b_k x^k$, the coefficients obey the estimates
$\abs{a_k}, \abs{b_k} \leq c_n M^n$. (This follows, for example,
by expressing $a_k$ and $b_k$ as sums of products of the corresponding roots.)

For  $\abs{x} \geq \max\{2M,1\}$, we now simply estimate 
\begin{align*}
\abs{\frac{P(x)}{Q(x)} - 1} &=
\abs{\frac{\sum_{k=0}^{n-1} (a_k-b_k) x^k}{Q(x)}}
\leq 
\frac{2^n  \sum_{k=0}^{n-1} \abs{a_k-b_k}\abs{x}^k}
{\abs{x}^n}
\leq \frac{2^{n+1} c_n M^n}{\abs{x}},
\end{align*}
and we may take $C_n = 2^{n+1} c_n$. 
\end{proof}

We then have the following completeness theorem. Note that it contains
Theorem~\ref{t:pol} as a special case, but the proof is more involved, so we presented the simpler case first.

\begin{tm} \label{t:rat}
  Let $g(x)= R(x)  e^{-\pi x^2}$ with a non-zero rational function  $R$
  that does not have any poles on  $\R $, and let  $\alpha \beta = p/q \in \bQ $, $\alpha \beta \leq 1$.   Then  the
Gabor system $\mathcal{G}(g, \alpha, \beta)$ is 
complete in $L^2(\mathbb{R})$.
\end{tm}

\begin{proof}
  The proof follows a similar outline as before for
  Theorem~\ref{t:pol}. In the proof of Theorem~\ref{t:pol} we
  evaluated $\expg$ in ~\eqref{sum2} for both $g(x)=P(x) e^{-\pi x^2}$
  and for the Gaussian $\phi (x) = e^{-\pi x^2}$ and then showed that
  the expression for the coefficient of the highest power in $x$
  equalled, up to an exponential factor, precisely the expression
  $\expphi$ for the Gaussian. Since $\expphi$ cannot vanish identically, 
neither can $\expg$ for $g = P  \phi $, whence the completeness of
  $\gab $. 

In the case of a function $g(x) = R(x) e^{-\pi x^2}$ with rational
$R$, we investigate the behavior of $\expg$ for $x \rightarrow \infty$.

Assume that $\gab $ is not complete. Then, by Lemma \ref{l:comsaa}
and the calculation in Step 1 of the proof of Theorem~\ref{t:pol}, we conclude that,
for all $x \in \Rst$ and $N \in \Zst$, 
\begin{align} 
\nonumber
0&= \expg 
\\
\nonumber
&= \sum_{k_0 + \dots+ k_{p-1} = N} \prod_{j=0}^{p-1} R\big(x+
\frac{\alpha j}{p} - \alpha k_j\big) \prod_{j=0}^{p-1} \phi \big(x+
\frac{\alpha j}{p} - \alpha k_j\big) c(k_0, \dots, k_{p-1})
\notag \\
&= e^{-\pi (px^2 + \alpha (p-1) x  - 2x \alpha N
  )} \sum_{k_0 + \dots+ k_{p-1} = N} \prod_{j=0}^{p-1} R\big(x+
\frac{\alpha j}{p} - \alpha k_j\big) \times  \notag \\
& \quad  \exp \Big(- \pi \alpha^2 \bigg[ \sum_{j=0}^{p-2} \bigg(j/p - k_j
\bigg)^2 +  \bigg(\frac{p-1}{p} -N+ \sum_{j=0}^{p-2}k_j
\bigg)^2\bigg] \Big) \, c(k_0, \dots, k_{p-1})\, . \label{eq:c24}
\end{align}
For $(k_0, \dots, k_{p-1}) \in \Zst^p$, let
\begin{align}
\label{eq_def_dk}
d(k_0, \dots, k_{p-1}) &:=
\exp \Big(- \pi \alpha^2 \sum_{j=0}^{p-1}
(j/p - k_j)^2\Big) c(k_0, \dots, k_{p-1}),
\end{align}
and drop one term in exponential in \eqref{eq:c24} to obtain
\begin{equation}
\label{eq:c24b}
\sum_{k_0 + \dots+ k_{p-1} = N} d(k_0, \dots, k_{p-1})
\prod_{j=0}^{p-1} R\big(x+
\frac{\alpha j}{p} - \alpha k_j\big)=0, \qquad x\in\Rst, N \in \Zst.
\end{equation}
Now let $R=P/Q$, with $P,Q \in \mathbb{C}[X]$ monic,
$\Delta = \mathrm{deg}\, (P) -   \mathrm{deg}\, (Q) \in \Zst$ and note
that 
\begin{equation}
  \label{eq:c23}
  \lim _{x\to \infty } \frac{P(x)}{x^\Delta Q(x)} = 1.
\end{equation}
After multiplying \eqref{eq:c24b} by $x^{-p\Delta }$, we obtain
\begin{align}
  \label{eq:c25}
0=   &\sum_{k_0 + \dots+ k_{p-1} = N} x^{-p \Delta } \prod_{j=0}^{p-1}  R\big(x+
\frac{\alpha j}{p} - \alpha k_j\big) d(k_0, \dots, k_{p-1}) \, .
\end{align}
Now we let $x$ tend to infinity, formally interchange the sum and the limit
- this is carefully justified below - and use \eqref{eq:c23} obtaining
\begin{align*}
  0 &=   \lim _{x\to \infty }  \sum_{k_0 + \dots+ k_{p-1} = N}
  x^{-p \Delta } \prod_{j=0}^{p-1}  R\big(x+
\frac{\alpha j}{p} - \alpha k_j\big) d(k_0, \dots, k_{p-1})  \\
&=    \sum_{k_0 + \dots+ k_{p-1} = N}  \lim _{x\to \infty }  x^{-p \Delta } \prod_{j=0}^{p-1} R\big(x+
\frac{\alpha j}{p} - \alpha k_j\big) d(k_0, \dots, k_{p-1}) \\
&=   \sum_{k_0 + \dots+ k_{p-1} = N}   d(k_0, \dots, k_{p-1})
= s_0(N),
\end{align*}
where the last expression is exactly $\expphi$ for the Gaussian $\phi$
(apart for an exponential term). Since $s_0(N)$ must be non-zero
for some $N$, we have arrived at a contradiction and conclude that
$\gab $ must be complete. 

For a rigorous proof we need to justify the interchange of the limit
and the sum. Let $M>1$ and $x \in \Rst$ be arbitrary.
Using 
\eqref{eq:c24b}, we write the partial sums of  $s_0(N)$
as:
\begin{align*}
s^M_0(N)&:=
\sum_{
\stackrel{k_0 + \dots+ k_{p-1} = N}{\abs{\newk}_2 \leq M}}
d(k_0, \dots, k_{p-1})
\\
&=
\sum_{
\stackrel{k_0 + \dots+ k_{p-1} = N}{\abs{\newk}_2 \leq M}}
d(k_0, \dots, k_{p-1})
\left(1-x^{-p\Delta}\prod_{j=0}^{p-1} R\big(x+
\frac{\alpha j}{p} - \alpha k_j\big)\right)+
\\
&\qquad
\sum_{
\stackrel{k_0 + \dots+ k_{p-1} = N}{\abs{\newk}_2 \leq M}}
d(k_0, \dots, k_{p-1})
x^{-p\Delta}\prod_{j=0}^{p-1} R\big(x+
\frac{\alpha j}{p} - \alpha k_j\big)
\\
&=
\sum_{
\stackrel{k_0 + \dots+ k_{p-1} = N}{\abs{\newk}_2 \leq M}}
d(k_0, \dots, k_{p-1})
\left(1-x^{-p\Delta}\prod_{j=0}^{p-1} R\big(x+
\frac{\alpha j}{p} - \alpha k_j\big)\right)-
\\
&\qquad -
\sum_{
\stackrel{k_0 + \dots+ k_{p-1} = N}{\abs{\newk}_2 > M}}
d(k_0, \dots, k_{p-1})
x^{-p\Delta}\prod_{j=0}^{p-1} R\big(x+
\frac{\alpha j}{p} - \alpha k_j\big)
\\
&=: t_N^{M}(x)-u_N^{M}(x).
\end{align*}
The  numbers $t^{M}_N(x)$ and  $u^{M}_N(x)$ depend on $x$, whereas the
partial sum $s_0^M(N)$ is independent of $x$. To show that $\lim
_{M\to \infty } 
s_0^M(N) = 0$, we will choose  $x= x_M$ judiciously  and 
obtain suitable bounds. First we note that the coefficients $d(k_0,\ldots,k_{p-1})$ satisfy the decay condition:
\begin{align}
\label{eq_dec_dk}
\abs{d(k_0,\ldots,k_{p-1})} \lesssim e^{-\gamma \abs{\newk}_2^2}
\end{align}
for some $\gamma>0$. Indeed, for each $j \in \{0,\ldots,p-1\}$ and $k_j \in \Zst$,
since $\abs{j/p} \leq \tfrac{p-1}{p} <1$, it follows that
$\abs{k_j - j/p} \gtrsim \abs{k_j}$.

We let $n:=\max\{\deg(P),\deg(Q)\}$. For $M \gg 1$ and
$\abs{\newk}_2 \leq M$, the function
\begin{align*}
x^{-p\Delta}\prod_{j=0}^{p-1} R\big(x+
\frac{\alpha j}{p} - \alpha k_j\big)
\end{align*}
is a quotient of two monic polynomials of degree $pn$ whose complex roots
lie inside a ball of radius $\lesssim M$. Using Lemma \ref{lemma_poly},
we now choose $x_M \in \Rst$ with $\abs{x_M} \asymp M^{n+1}$ such that
\begin{align*}
\abs{1-x_M^{-p\Delta}\prod_{j=0}^{p-1} R\big(x_M+
\frac{\alpha j}{p} - \alpha k_j\big)} \lesssim \frac{1}{M}.
\end{align*}
Combining this estimate with \eqref{eq_dec_dk} we obtain
\begin{align}
\label{eq_tm}
\abs{t^{M}_N(x_M)} \lesssim \frac{1}{M}
\sum_{\newk} \abs{d_{\newk}}
\lesssim \frac{1}{M}.
\end{align}
Second, since $R$ is a rational function without real poles, it satisfies the estimate
\begin{align*}
\abs{R(x)} \lesssim (1+\abs{x})^{l}, \qquad x \in \Rst,
\end{align*}
for $l:=\max\{\Delta,0\}$. This allows us to bound,
for $\newk \not= 0$,
\begin{align*}
&\abs{x_M^{-p\Delta}\prod_{j=0}^{p-1}R\big(x_M+
\frac{\alpha j}{p} - \alpha k_j\big)}
\lesssim M^{(n+1)p\abs{\Delta}}
\prod_{j=0}^{p-1}
(1+|x_M+\frac{\alpha j}{p} - \alpha k_j|)^{l}
\\
&\qquad
\lesssim M^{(n+1)p\abs{\Delta}}
\left(
M^{(n+1)lp}+\abs{\newk}_2^{lp}\right)
\lesssim
M^s \abs{\newk}_2^s,
\end{align*}
for some $s>0$. Therefore
\begin{align*}
\abs{u^{M}_N(x_M)} \lesssim
M^s \sum_{\abs{\newk}_2>M}
\abs{\newk}_2^s e^{-\gamma\abs{\newk}^2_2}.
\end{align*}
The last bound shows that $\abs{u^{M}_N(x_M)} \longrightarrow 0$
as $M \longrightarrow +\infty$. Combining this with \eqref{eq_tm} we conclude that
\begin{align*}
s_0(N)=\lim_{M \longrightarrow +\infty} s^M_0(N)=0,
\end{align*}
and  the proof is complete.
\end{proof}

As a consequence of Theorem \ref{t:rat}, we obtain the completeness of Gabor systems for a
class of  totally positive functions. Let $\delta _j \in \bR $, $j=1,
\dots , M$, $\gamma >0$ and set 
\begin{equation}
  \label{eq:c26} 
\hat{g}(\xi ) = e^{-\gamma \xi ^2} \prod _{j=1}^M (1+2\pi  i \delta _j
\xi )\inv \, .
\end{equation}
Schoenberg's factorization theorem \cite{schoenberg1947totally} asserts that $g$ is a
totally positive function. Since $\hat{g}$ satisfies the
assumptions of Theorem~\ref{t:rat} and the completeness of a Gabor
system is invariant under the Fourier transform, we obtain the following corollary.

 \begin{cor} \label{cortp}
Let $g$ be a totally positive function whose Fourier transform factors
as in \eqref{eq:c26} and assume that $\alpha \beta \leq 1$ is
rational. Then $\gab $ is complete in $\ltwo $.  
\end{cor} 
To put the corollary into perspective, we note that it is conjectured
that for an  arbitrary  totally positive function $g$  the Gabor system
$\gab $ is a frame, \fif\ $\alpha \beta < 1$. At this time, this
conjecture is known to be true only for totally positive functions
whose  Fourier transform factors as  $\prod _{j=1}^M (1+2\pi  i \delta _j
\xi )\inv $~\cite{GS13}. Corollary \ref{cortp} shows at least
completeness for another class of totally positive functions.

As a final application of the method of proof used in Theorem
\ref{t:pol} we treat  windows of the type 
$E(x)\phi(x)$, where $E$ is an exponential polynomial 
\begin{equation} \label{expsum}
E(x)= \sum_{\lambda \in \Lambda} a_\lambda e^{\lambda x},
\end{equation}
for some  finite set $\Lambda \subset \bC $ and non-zero coefficients  $a_\lambda \in \bC$.
\begin{tm} \label{t:expsum}
Let $g(x)=E(x)\phi(x)$ where $E(x)$ is of the form \eqref{expsum},
$\alpha \beta \leq 1$, and $\alpha \beta \in \bQ$.  Then $\gab$ is
complete in $L^2(\bR )$.
\end{tm}
\begin{proof}
We proceed  as in the proofs of Theorems \ref{t:pol}
and \ref{t:rat}. Writing $\alpha \beta =p/q$ and with the notation from \eqref{eq_def_dk},
the incompleteness condition of Lemma \ref{l:comsaa} becomes
\begin{multline} \label{horriblesum}
0=
\sum_{\lambda_1, \dots, \lambda_p \in \Lambda} a_{\lambda_0} \cdots a_{\lambda_{p-1}} 
e^{\sum_{j=0}^{p-1} \lambda_j x} \sum_{k_0+\dots+k_{p-1}=N} e^{-\alpha \sum_{j=0}^{p-1} \lambda_j k_j}
d(k_0, \dots, k_{p-1}),
\end{multline}
for all $N \in \Zst$. Next, let $\lambda ^*\in \Lambda $ be the
minimal element of $\Lambda $  in the lexicographic ordering of $\bR
^2$. In other words, let $\tilde \Lambda$ be the subset 
of $\Lambda$ which contains the elements of $\Lambda$ with  minimal
real part and  let $\lambda^*$ be the 
element of $\tilde{\Lambda}$ with minimal  imaginary part. 
Then $\sum_{j=0}^{p-1} \lambda_j = p \lambda^*$ if and only if 
$\lambda_1=\dots=\lambda_p=\lambda^*$.
Since the complex exponentials $x \to e^{wx}, w\in \bC, x\in \bR$, are
linearly independent, the coefficient of $e^{p\lambda ^*x} $ must
vanish. Since $a_\lambda \neq 0$ for all $\lambda \in \Lambda $, 
it follows from \eqref{horriblesum} that
\begin{align*}
0= e^{-\alpha Np \lambda^*} \sum_{k_0+\dots+k_{p-1}=N} 
d(k_0, \dots, k_{p-1}) = e^{-\alpha Np \lambda^*} s_0(N),
\end{align*} 
for all $N \in \Zst$. As in the proofs of Theorems \ref{t:pol} and Theorem \ref{t:rat},
this contradicts the completeness of $\mathcal{G}(\phi,\alpha,\beta)$, with $\phi$ the Gaussian.
\end{proof}

\begin{cor}
  Let $g = \sum _{j=1}^n d_j M_{b _j} T_{a_j} \phi $ be a non-zero
  finite linear combination of time-frequency shifts of the Gaussian 
$\phi$, $\alpha \beta \leq 1$, and $\alpha\beta \in \bQ$. Then 
$\gab $
is complete in $L^2(\bR )$. 
\end{cor}

\begin{proof}
  This is clear since sums of \tfs s can be written in the
  form~\eqref{expsum}, as $M_b T_a \phi (x) = e^{2\pi \lambda x}
  \phi (x) e^{-\pi b ^2}$ for $a,b \in \bR $ and $\lambda = a
  +ib $. 
\end{proof}

\textbf{Concluding remarks:}
1. In view of Theorem~\ref{tm:intro} it is tempting to conjecture that
$\gab $ with $\alpha \beta \leq 1$ is complete for \emph{every}
function 
$g\in \saa $. If true, the proof of this statement must heavily depend
on the analyticity of $g$ and $\hat{g}$ and its Zak transform
$Z_\alpha g$. The following simple counter-example shows what may
happen without analyticity. Let  $g $ be  a Schwartz function  with support in $\bigcup _{j\in \bZ } [2j, 2j+1]$, then $\gab
$ is incomplete  whenever $\alpha \in [0,1)+2\bZ $ and  $\beta >0$,
although $g$ and $\hat{g}$ are $C^\infty $. 

2.  Furthermore we remark that the Gabor systems $\gab $ with $\alpha
\beta <1$ in
Theorem~\ref{tm:intro}, although being complete, cannot be   Schauder
bases.  A  Gabor 
Schauder basis $\gab $ must satisfy  $\alpha \beta = 1$ and the
corresponding  window $g$ is  poorly 
localized either in the time or the frequency domain \cite{dehe00,
  MR2278667}. 

3. The completeness results depend heavily on the lattice structure of
the phase space shifts. Although the examples seem counter-intuitive
from the point of view of quantum mechanics, 
 one may construct
complete Gabor systems  without lattice structure that are  of density zero~\cite{ALS09}.

\vs 

\textbf{Acknowledgements:} The second author would like to thank
Kristian Seip and Yurii Lyubarskii for useful discussions.

\def\cprime{$'$} \def\cprime{$'$} \def\cprime{$'$} \def\cprime{$'$}
  \def\cprime{$'$} \def\cprime{$'$}

\bigskip

\end{document}